\documentclass[journal,draftcls,onecolumn,12pt,twoside]{IEEEtranTCOM}
\ifCLASSINFOpdf
  
\else
  
\fi

\usepackage{amsmath}
\usepackage{graphicx}
\usepackage{amssymb}
\usepackage{amsthm}
\usepackage{cite}
\usepackage{mathtools}
\usepackage{bbold}
\usepackage{steinmetz}
\usepackage{stmaryrd}
\usepackage{bbm, dsfont}
\usepackage{verbatim}
\usepackage[latin1]{inputenc}
\usepackage{tikz}
\usetikzlibrary{shapes,arrows}
\usepackage{enumerate}
\usepackage{color}
\usepackage[caption=false,font=footnotesize]{subfig}
\usepackage{pgfplots}
\usetikzlibrary{calc}
\usepackage{mathrsfs}
\usepackage{slashbox}
\usepackage{steinmetz}

\usepackage{lipsum}
\usepackage[variablett]{lmodern}
\usepackage{amsbsy}
\usepackage{bm}
\usetikzlibrary{positioning}
\usetikzlibrary{shapes,arrows}
\usepackage{multirow}

\pgfplotsset{
	every tick label/.append style={scale=1},
	every axis/.append style={
	}
}

\pgfplotsset{
	grid style = {
		dash pattern = on 0.05mm off 1mm,
		line cap = round,
		black,
		line width = 0.5pt
	}
}

\newcommand{\red}{\Re} 
\newcommand{\ima}{\Im}                                                

\newcommand{\df}{\Delta f}

\newcommand{\leff}{L_{\mathrm{eff}}}
\newcommand{\ns}{n_{\mathrm{span}}}

\newcommand{\hst}[2]{\tilde{h}_{#1}(#2)}
\newcommand{\hs}[2]{h_{#1}(#2)}
\newcommand{\Aa}[3]{\mathbf{a}_{#1}(#2,#3)}
\newcommand{\Ab}[2]{\mathbf{a}(#1,#2)}
\newcommand{\Aad}[1]{\mathbf{a}_{#1}}
\newcommand{\Abd}{\mathbf{a}}
\newcommand{\Aar}[3]{\red\lefto(\mathbf{a}_{#1}(#2,#3)\right)}
\newcommand{\Aai}[3]{\ima\lefto(\mathbf{a}_{#1}(#2,#3)\right)}
\newcommand{\Nn}[2]{\mathbf{n}_{#1}(#2)}
\newcommand{\Nnr}[2]{\red\rlpo{\mathbf{n}_{#1}(#2)}}
\newcommand{\Nni}[2]{\ima\rlpo{\mathbf{n}_{#1}(#2)}}
\newcommand{\phii}[2]{\phi(#1 , #2)}

\newcommand{\si}{\mathbf s}
\newcommand{\din}[2]{\mathbf x_{#1#2}}
\newcommand{\dii}[1]{\mathbf x_{#1}}
\newcommand{\diip}[1]{\mathbf x'_{#1}}
\newcommand{\diin}[1]{ x_{#1}}

\newcommand{\vU}{\ul{\mathbf{u}}}
\newcommand{\vUs}{\ul{{u}}}
\newcommand{\FFIIc}[1]{\tilde{\Phi}(#1)}
\newcommand{\FFIIs}[1]{\Phi(#1)}

\newcommand{\sch}{Schr\"odinger}

\newcommand{\omfs}{{\mathbf v}}
\newcommand{\omfss}{{ v}}
\newcommand{\omtd}{{\mathbf w}}

\newcommand{\mf}{\mathbf}                                                
\newcommand{\chek}[1]{\quad\href{run:./references/#1}{\checkmark}}

\newcommand{\expe}{\mathrm E}

\newcommand{\bo}{\mathcal{O}}

\newcommand{\ul}[1]{\underline{#1}}                  
                  
\newcommand{\be}{\begin{equation}}
	\newcommand{\ee}{\end{equation}}
\newcommand{\bs}{\begin{split}}
	\newcommand{\es}{\end{split}}

\newtheorem{theorem}{Theorem}

\newtheorem{proposition}{Proposition}

\newcommand{\bal}{\begin{align}}
	
	\newcommand{\eal}{\end{align}}

\newcommand{\rlp}[1]{\left(#1\right)}

\newcommand{\rlb}[1]{\left[#1\right]}
\newcommand{\rlpo}[1]{\lefto(#1\right)}

\newcommand{\rlbo}[1]{\lefto[#1\right]}

\newcommand{\lefto}{\mathopen{}\left}
\newcommand{\distas}{\sim}					

\newcommand{\di}{\mathop{}\!\mathrm{d}}

\begin{document}

\title{{  Demodulation and Detection Schemes for a  Memoryless  Optical WDM Channel}
}

\author{
	Kamran~Keykhosravi,~\IEEEmembership{Student~Member,~IEEE,}
	Morteza~Tavana,~\IEEEmembership{Student~Member,~IEEE,}\\
	Erik~Agrell,~\IEEEmembership{Fellow,~IEEE,}\\
	and~Giuseppe~Durisi,~\IEEEmembership{Senior~Member,~IEEE.}
	
	\thanks{Parts of this paper  have been  presented at the IEEE International Symposium on Information Theory (ISIT), June~2017, Aachen, Germany \cite{Kamran17_ISIT}.  }
\thanks{ This work 
	was supported by the Swedish Research Council (VR) under Grant 2013-5271 and by the European commission under Grant MSCA-ITN-EID-676448. }
\thanks{The authors are with the Department of Electrical Engineering,
	Chalmers University of Technology, Gothenburg 41296, Sweden (e-mail:
	kamrank@chalmers.se; tavanam@chalmers.se; agrell@chalmers.se; durisi@chalmers.se).}	
}

\maketitle

\begin{abstract}
It is well known that  matched filtering and sampling (MFS) demodulation together with   minimum Euclidean distance (MD) detection constitute the optimal receiver for the additive white Gaussian noise channel. However,  for a general nonlinear transmission medium, MFS does not provide sufficient statistics, and therefore is suboptimal. Nonetheless, this receiver is widely used in optical systems, where the Kerr nonlinearity is the dominant impairment at  high powers. 
In this paper, we consider a suite of receivers  for a  two-user channel subject to a type of nonlinear interference that occurs  in wavelength-division-multiplexed  channels.
{  The asymptotes of the symbol error rate (SER)  of the considered receivers at high powers  are derived or bounded analytically. Moreover, }
Monte-Carlo simulations are conducted to evaluate the SER  for all the  receivers. 
Our results show that  receivers that are based on MFS cannot achieve arbitrary low SERs, whereas the SER goes to zero as the power grows for the optimal receiver. Furthermore, we devise a  heuristic demodulator, which together with the MD detector yields a receiver that is simpler than the optimal one and can achieve arbitrary low SERs. { The SER performance of the proposed receivers is also evaluated for some  single-span fiber-optical channels via split-step Fourier simulations.

 }

\end{abstract}

\begin{IEEEkeywords}
	Optical fiber, nonlinearity compensation, nonlinear channel,  demodulation, MAP detector.
\end{IEEEkeywords}

\IEEEpeerreviewmaketitle

\section{Introduction}\label{s1}

The development of the standard single-mode fiber {(SMF)} in the $1970$s \cite{miya1979ultimate}  and of the erbium-doped fiber amplifiers \cite{mears1987low}  in the late $80$s increased the capacity of the fiber-optical channel  far beyond the required data rate in those days.
This abundance of resources made it inessential to exploit the bandwidth optimally in the design of  optical communication networks. 
Nowadays, however,  with the exponential growth of the global Internet, the data demand has started meeting the limits of  traditional optical systems. 
This ever-increasing data demand has motivated many  recent efforts, including the one in the current paper, to increase the efficiency of  optical transmitters and receivers.

For the additive white Gaussian noise (AWGN) channel, it is well known that the matched filtering and sampling (MFS) demodulator provides  sufficient statistics for detecting the transmitted symbol from the received continuous-time signal.
Although, in general, this method is suboptimal for nonlinear channels, it has been deployed broadly in optical fiber transmission systems,  where the  Kerr nonlinearity critically limits the achievable information rate at moderate and high powers \cite{essiambre_2010_jlt}.

In  advanced optical communication systems, a single-mode fiber  hosts approximately one hundred wavelength-division-multiplexed (WDM) channels.
 In such systems, the  Kerr nonlinearity gives rise to: \textit{i)} self-phase modulation (SPM), where the signal phase is distorted depending on its own magnitude; \textit{ii)}  cross-phase modulation (XPM), where the magnitude of the signal transmitted over  neighboring channels modulates the phase of the signal of interest; and \textit{iii)} four-wave mixing (FWM), where three signals at different frequencies create a distortion at a new frequency.
 In this paper, we shall  focus  on the first two effects and assume that the impact of FWM (the third effect) is mitigated by  appropriate channel spacing (see, for example, \cite{forghieri1994reduction}).

Many methods have been proposed, both in the optical and the electrical domains, to compensate for the fiber nonlinear distortion \cite[Ch.~2]{naga_2016_phd}.
Soliton-based communication \cite{lax1968integrals} is among the primary solutions to mitigate the channel impairments including the nonlinearity. It is based on soliton pulses, which can propagate through the fiber undisturbed. In  recent years, this method has received attention in the context of the nonlinear Fourier transform \cite{yousefi_2014_tinf3}. Inverting the signal's phase at the middle of the transmission line is another effective approach to reduce the nonlinear distortion \cite{fisher1983optical}.

In the last decade, the advancement of digital signal processors (DSP) made them a key enabling technology for data transmission over the fiber-optical channel. A number of known nonlinearity mitigation techniques are based on DSPs, three of which are  reviewed next. \textit{i)} \textit{Digital back propagation} \cite{ip_2008_jlt}  is a well-known method to compensate for the fiber impairments. Using this technique, all the signal--signal distortions can be compensated for by processing the signal at the transmitter, at the receiver, or at both ends. However, digital back propagation suffers from high computational complexity, and it requires knowledge of all  copropagating channels. \textit{ii)} The effects of XPM can be partially mitigated via \textit{adaptive equalization} that utilizes the time coherency of the XPM distortions {(see for example  \cite{secondini2014xpm,arlunno2011digital})}. 
\textit{iii)} Using an approximate probability distribution for the  {channel law}, one can devise  \textit{nonlinearity-tailored detection} techniques to improve the symbol error rate (SER) \cite{marsella2014maximum,cai2010experimental}.

The optical channel model can be described by the nonlinear \sch{}  (NLS) equation \cite[Eq.~2.6.18]{agrawal_2002_focs}. Since  the input--output relation is given implicitly through a differential equation, developing the optimal transmitter and receiver for the NLS channel  seems a formidable task. 
{ 
By neglecting the channel memory, closed-form input--output relations can be obtained. } The analyses based on these models are applicable to optical systems with { short-haul} zero-dispersion fibers (see, for example {\cite{tan_2011_icc,lau_2007_jlt,keykhosravi-ecoc-17,yousefi_2011_tinf,beygi_2011_jlt}}). 
Furthermore, many simplified  models have been developed in the literature to approximate the NLS channel (see \cite{Erik_IT_friendly} and the references therein). 
Applying perturbation theory, { or equivalently Volterra series,} and ignoring signal--noise interaction are among the most common simplifications. 
The channel models derived based on these assumptions lose accuracy at high powers \cite{keykhosravi-tit-arxiv-2017}.
Nevertheless, since the physical channel is intractable, these models can be studied to develop transceivers that are more matched to the nonlinear nature of the optical channel than the MFS. The corresponding results can serve as a first step towards optimizing  optical receivers for the actual physical channel.

  In \cite{ghozlan2010interference,ghozlan2011interference}, the capacity of a  memoryless discrete-time two-user WDM channel, where both SPM and XPM are present,  has been studied at high powers. It has been proved for this channel that the capacity pre-log\footnote{The capacity pre-log is defined by $\lim\limits_{P\to\infty}\mathcal{C}(P)/\log P$, where $\mathcal{C}(P)$ is the channel capacity  under the input power constraint $P$.} pair $(1,1)$ is achievable.
  The discrete-time channel model used in \cite{ghozlan2010interference,ghozlan2011interference} relies on  the sampling receiver, { whose bandwidth is infinite. }
  This receiver has been used in many publications to obtain a tractable discrete-time model for the single-user NLS channel (see, for example \cite{hager2013design,yousefi_2011_tinf,keykhosravi-tit-arxiv-2017,beygi_2011_jlt}). However, the sampling receiver is suboptimal and  impractical, particularly for WDM systems  \cite[Sec.~I]{kramer2017autocorrelation}. 
  The discrete-time channel in \cite{ghozlan2010interference,ghozlan2011interference} can also be obtained from the underlying continuous-time channel by using  rectangular pulse shaping at the modulator, which, however, cannot be  implemented in practice.

 This paper studies the same continuous-time two-user WDM channel as in \cite{ghozlan2010interference,ghozlan2011interference}.
 Although our focus in this paper is on a two-user channel, our framework can be used to analyze  a channel of interest in a   WDM system with an arbitrary number of users by considering   all of the interfering signals as a single channel \cite{secondini_2013_jlt,secondini2012analytical}.
 We consider three  demodulation  schemes for the aforementioned continuous-time channel under the assumption that joint processing is not possible at the transmitters or at the receivers. 
 First, the MFS demodulator is studied, which is conventionally used in optical systems. 
 Second, a demodulator that provides sufficient statistics (SS) is developed.
  Third, a novel heuristic demodulation method, referred to as maximum matching (MxM), is presented.
   Furthermore, three different detection schemes,  used at the receivers to estimate the transmitted signal based on the demodulator output,  are considered: the conventional minimum Euclidean distance (MD) detector, the optimal  detector based on maximum a posteriori (MAP) probability, and   a two-stage (TS) detection method, which first estimates the amplitude and then the phase of each symbol. Different versions of TS detectors have been  considered  previously to mitigate the nonlinear phase noise in optical systems \cite{lau_2007_jlt,hager2013design,beygi2011optimization}. As we shall see, our TS detector is superior to the MD detector  at moderate powers.

   \begin{table}[!t]
   	\centering
   	\caption{A qualitative comparison between the complexity and performance of the receivers under study.}
   	\label{table1}
   	\begin{tabular}{c|c|c||c|c|c|}
   		\cline{2-6}
   		& \multicolumn{2}{|c||}{Complexity} &\multicolumn{3}{|c|}{Symbol error rate compared to the optimal receiver}\\
   		\cline{1-6}
   		\multicolumn{1}{|c||}{Receiver}& Demodulation & Detection & Low powers & Moderate powers & High powers  \\
   		\hline\hline
   		\multicolumn{1}{|c||}{MFS-MD} & Low&  Low& Close to optimal & Far from optimal& Far from optimal \\
   		\hline
   		\multicolumn{1}{|c||}{MFS-PR} & Low&  Low& Far from optimal & Close to optimal& Far from optimal \\
   		\hline
   		\multicolumn{1}{|c||}{MFS-MAP} & Low& High & Close to optimal &Close to optimal & Far from optimal \\
   		\hline
   		\multicolumn{1}{|c||}{SS-MAP} & High& High & Optimal &Optimal &Optimal ($\to 0$)$^*$ \\
   		\hline
   		\multicolumn{1}{|c||}{MxM-MD} & High&  Low& Close to optimal &Far from optimal & Close to optimal ($\to 0$) \\
   		\hline
   		\multicolumn{1}{|c||}{MxM-TS} &High & Low & Far from optimal & Far from optimal$^\dagger$  & Close to optimal ($\to 0$)\\
   		\hline
   		\multicolumn{6}{l}{* SER$\to 0$ as power grows large. \qquad $^\dagger$ The SER with  MxM-TS, is lower than with MxM-MD at moderate powers.}\\
   	\end{tabular}
   	
   \end{table}

By coupling different modulators and detectors, we investigate the performance (in terms of SER) and the complexity of {six} different receivers.
   First, we study the conventional MFS-MD receiver, which is optimal for the linear AWGN channel.
   { Second, we study  a receiver that performs MFS demodulation,  phase recovery (using the method in  \cite{pfau_2009_jlt}), and MD detection. This receiver, which relies on processing techniques used in today's optical systems, is referred to as MFS-PR.}
   Third, to find the performance limits of the MFS demodulator, we couple it with the optimal (MAP) detector.
   Fourth, we consider  the SS-MAP receiver, which is the optimal receiver for the channel under study.
   Fifth, we couple  MxM with MD to obtain a receiver that has a lower complexity than SS-MAP and can achieve arbitrary low SERs.
   Sixth, we study the MxM-TS receiver, which turns out to yield a slight  performance improvement over MxM-MD at moderate powers.
   %
   %
    A summary of  the considered receivers and a qualitative evaluation of their complexity  and performance is provided in Table~\ref{table1}. 
    At low powers, where nonlinearity is weak, all the receivers except the MxM-TS { and MFS-PR} have approximately the same SER as the optimal receiver, whereas in the moderate-power regime only MFS-MAP { and MFS-PR} perform close to optimal.
    It can be seen that unlike receivers based on MFS, the SER for the optimal receiver (SS-MAP) goes to zero as the power grows large.
   Also, arbitrarily low SERs can be achieved via simple detectors (MD and TS) coupled with the MxM demodulator. 
   { The results presented in Table~\ref{table1} are obtained for  truncated Gaussian pulse shaping and 16-QAM modulation. We expect similar results to hold for  practically relevant pulse shapes whose spectrum broadens with increasing power (see \cite[Sec.~VIII]{yousefi_2011_tinf}). For  rectangular pulse shaping, for which the signal spectrum does not broaden,  MFS provides  sufficient statistics and the SER of MFS-MAP goes to zero as power grows large. Modulation formats that are resilient to phase noise, such as pulse-amplitude modulation, may also result in a different SNR behavior compared to Table~\ref{table1}. 
      } 

{
We also evaluate the SER performance of the proposed receivers (by means of split-step Fourier simulations) for two single-span fiber-optical systems with different dispersion parameters.   Our results show that, for all receivers, the SER increases with power after a certain optimal power.  When the dispersion is small, the performance of SS-MAP and MFS-MAP turns out  to be superior to that of MFS-PR. When dispersion is high,   all  receivers except for SS-MAP are  inferior to MFS-PR. }       This paper  completes the analysis initiated with the conference paper \cite{Kamran17_ISIT}, where the MxM-MD and the MFS-MD receivers were investigated for the channel under study.

\textit{Organization:} The rest of this paper is organized as follows. In Section~\ref{s2}, a model for a continuous-time two-user WDM channel is obtained from a pair of coupled NLS equations under some simplifications. In Section~\ref{s3}, we present the  demodulation and detection methods.  
{ Section~\ref{siv} presents some analytical asymptotic bounds on the SER. } Numerical  results  are provided in Section~\ref{s4}. 
{ Specifically, in Section~\ref{s41}, we study the simplified channel model  and in Section~\ref{s5.5}, the performance under more realistic dispersive conditions is evaluated by simulation.}
Finally, Section~\ref{s5} concludes the paper.

\textit{Notation:}  Bold-face letters are used to denote random quantities. 
Sets are indicated by upper-case script letters, e.g., $\mathcal{X}$. The cardinality of a set $\mathcal{X}$ is indicated by $|\mathcal{X}|$. 
Vectors are denoted by lower-case underlined letters.
 $\mathcal{CN}(\mu, \sigma^2)$ denotes the proper complex  Gaussian distribution with mean $\mu$ and variance $\sigma^2$. The inner product between two complex functions $f(t)$ and $g(t)$ is defined as $\langle f,g\rangle=\int_{-\infty}^{\infty}f(t)g^*(t)\di t$, where $(\cdot)^*$ denotes  complex conjugation.
 $\red(x)$ and $\ima(x)$ denote the real and the imaginary part of a complex number $x$, respectively.
With $|\cdot|$ and $(\cdot)^T$ we denote the determinant and the transpose operators, respectively.
 We use $\mathrm{Pr}(\mathbf x=x)$  to denote the probability mass  function of a discrete random variable $\mf x$ at  $x$. Also, the probability density function of a continuous random variable $\mf x$ at  $x$ is denoted by $f_{\mathbf x}(x)$. 
The real line and the complex plane are represented by $\mathbb R$ and $\mathbb C$, respectively. { Finally, for two functions $q(x)$ and $r(x)$, we write $q(x)=\bo\rlpo{r(x)}$ if $\limsup\limits_{x\to 0}|q(x)/r(x)|<\infty$.}

\section{Channel Model}\label{s2}
The signal propagation through the fiber-optical  channel suffers from several impairments such as chromatic dispersion, fiber loss, and  Kerr nonlinearity.
The chromatic dispersion is mainly caused by the dependency of the refractive index on the frequency. 
Therefore, in the presence of chromatic dispersion, the  different frequency components of a transmitted pulse propagate with different speeds, causing the pulse to  broaden in time. This impairment can be compensated for by using dispersion-compensating fibers or through DSPs.

To compensate for the fiber loss, two types of optical amplification are typically deployed, namely, distributed or lumped amplification.
While the former amplifies the signal continuously during propagation, the latter does so only at the end of each amplification span. Optical amplification is always  accompanied by  additive noise caused by spontaneously emitted light photons.  In this paper, we shall focus on  lumped-amplified  systems.

The main impairment that limits the achievable data rates in fiber communications is the Kerr nonlinearity. 
It arises because the glass refractive index depends on the propagating optical power. 
It can be described by a phase shift proportional to the optical power applied to the complex baseband signal.
This phase shift is caused by the signal itself (SPM) or by other copropagating signals at different wavelengths (XPM).

In this paper, { we consider two channel models: a simple memoryless model for algorithm design and analysis, and a more realistic split-step Fourier model for performance evaluation. For the first purpose,} we consider the propagation of two optical signals with different carrier wavelengths through a point-to-point single-mode fiber, focusing on the effects of { SPM and } XPM. 
We  assume that the two signals have nonoverlapping spectra. 
The signal propagation  can then be described by the  pair of coupled NLS equations 
\cite[Eqs.~(7.4.1)--(7.4.2)]{agrawal_2007_nfo}
\begin{align}
	\frac{\partial\Aad{1}}{\partial z}+\frac{j\beta_{21}}{2}\frac{\partial \Aad{1}}{\partial t^2}+\frac{\alpha}{2}\Aad{1}
	&=j\gamma_1\rlp{|\Aad{1}|^2+2|\Aad{2}|^2}\Aad{1}\label{ch:nlse:1}\\
		\frac{\partial\Aad{2}}{\partial z}+d\frac{\partial \Aad{2}}{\partial t}+\frac{j\beta_{22}}{2}\frac{\partial \Aad{2}}{\partial t^2}+\frac{\alpha}{2}\Aad{2}
	&=j\gamma_2\rlp{|\Aad{2}|^2+2|\Aad{1}|^2}\Aad{2}\label{ch:nlse:2}
	\end{align}
where $\Aad{k}=\Aa{k}{z}{t},\ k\in\{1,2\}$ is the complex envelope of the optical signal $k$ at position $z$ and time $t$. Time is measured according to a reference frame moving with $\Aa{1}{z}{t}$. The group-velocity mismatch between the two channels is given by $d$. The constants  $\beta_{2k}$ and $\gamma_k$ are the  dispersion and the nonlinearity coefficients, respectively.
The fiber loss, which is assumed to be the same in both channels, is quantified by the parameter $\alpha$. { Although our focus in this paper is on single-polarization transmission, our analytical framework can be adapted to  suit an extension of the  channel model \eqref{ch:nlse:1}--\eqref{ch:nlse:2}  to two polarizations (see \cite[Eqs.~(7.1.19)--(7.1.20)]{agrawal_2001_nfo}).}

We assume that the fiber loss is completely compensated for using lumped amplification and that each amplifier generates Gaussian noise. { Moreover,  we assume that the effects of dispersion,  group velocity mismatch, and signal--noise interaction are negligible. This assumption is valid for  single-span short-haul communication systems with (optical or digital) dispersion compensation.  
 Under this assumption, the} coupled NLS equations  \eqref{ch:nlse:1}--\eqref{ch:nlse:2} yield  the continuous-time channel 
\cite[Eq.~(7.4.5)]{agrawal_2007_nfo}
\begin{IEEEeqnarray}{rCl}
	\Aa{1}{L}{t}&=&\Aa{1}{0}{t}e^{j\eta_1\rlp{|\Aa{1}{0}{t}|^2+2|\Aa{2}{0}{t}|^2}}+\Nn{1}{t}\label{cont_1}\\
	\Aa{2}{L}{t}&=&\Aa{2}{0}{t}e^{j\eta_2\rlp{|\Aa{2}{0}{t}|^2+2|\Aa{1}{0}{t}|^2}}+\Nn{2}{t}\label{cont_2}.
\end{IEEEeqnarray}
Here, $L$ is the length of the fiber. The parameters $\eta_k$ quantify the nonlinearity and can be calculated as
\begin{equation}\label{eta}
\eta_k=\ns\gamma_k\leff
\end{equation}
where $\ns$ is the number of amplification spans and
\begin{equation}\label{leff}
\leff=\frac{1-e^{-\alpha L_{\mathrm{span}}}}{\alpha}
\end{equation}
is the effective length of the fiber in a single span with length $L_\mathrm{span}=L/n_{\mathrm{span}}$. Because of fiber loss, the signal
power and, consequently, the nonlinear distortion, diminishes along the fiber. Therefore, the effective length is less
than the actual span length $L_{\mathrm{span}}.$
 Finally, the amplification noise is captured by $\Nn{1}{t}$ and $\Nn{2}{t}$, which are two independent  complex white circularly-symmetric Gaussian   processes with {power} spectral density
\begin{equation}\label{eq:noise:var}
	N_0=\frac{1}{2}n_{\mathrm{span}} h\nu FG.
\end{equation}
Here, $h \nu$ is the optical photon energy, $F$ is the noise figure, and  $G$ is the amplifier gain, which we assume equal to the signal attenuation in one span $\exp(\alpha L_{\mathrm{span}})$. 

In this paper, we shall {first}  focus on the simplified continuous-time model \eqref{cont_1}--\eqref{cont_2} and study the SER performance of different demodulation and detection schemes. { A more realistic channel model is studied in Section~\ref{s5.5}}.
Throughout the paper, we assume that the  parameters of the fiber are known at both receivers.
Moreover,  we assume  that the messages sent over each channel are  independent, and that joint  processing is not allowed at the transmitters or  receivers.


\section{Modulation, Demodulation, and Detection }\label{s3}
 In this section, a modulation scheme together with the six receivers listed in Table~\ref{table1} are presented for the continuous-time channel  \eqref{cont_1}--\eqref{cont_2}.
 The transmitters are assumed to perform linear modulation.
 Specifically, let the pulse shape $g(t)$ be a real function that is zero outside  the interval $(0,T]$ and has  unit energy, i.e., $\int_{0}^{T}g^2(t)\di t=1$. 
 Furthermore, define $\Aa{k}{0}{t} = \sum_i \din{k}{i}g(t-iT)$ to be the signal sent by  transmitter $k$, where $\din{k}{i}\in\mathbb C$ is the $i$th transmitted symbol. 
{ 	 Since $g(t)$ is zero outside $(0,T]$, after demodulation the noise terms at different symbol times become independent. Based on this and  the fact that the channel model is memoryless, the channel can be studied by only considering the input--output relation in the first symbol interval.}  Hence, we can drop the index $i$.
 The channel \eqref{cont_1} can be expressed as
 \begin{IEEEeqnarray}{rCl}\label{ch}
 	\Aa{1}{L}{t}&=& \dii{1}g(t)\exp\rlpo{j\eta_1\rlp{|\dii{1}{}|^2+2|\dii{2}{}|^2}g^2(t)}+\Nn{1}{t}  \quad 0\leq t< T
 \end{IEEEeqnarray}
 where  we set $\dii{k}=\din{k}{1}$ for $k=1, 2$ to simplify notation. 
 In this section, we focus only on the first WDM channel \eqref{cont_1}. Because of the symmetry, all the results hold for the second channel \eqref{cont_2} as well.
 
 Next, we introduce some notation that will come to use in the rest of this section.
We assume that  the input random variable $\dii{1}$ takes values from a finite-cardinality set $\mathcal{X}=\{x_1, x_2,..., x_{|\mathcal{X}|} \}$ and has a probability distribution $\pi_i=\mathrm{Pr}(\dii{1}=x_i)$.
 Furthermore, we assume that $\dii{2}$ belongs to a finite-cardinality set, which may be different from $\mathcal{X}$.
Also,  we let $\si=|\dii{1}|^2+2|\dii{2}|^2$, which belongs to  a finite set $\mathcal{S}=\{s_1, s_2,..., s_{|\mathcal{S}|} \}$.  
Finally,  we denote  the  conditional probability distribution of $\si$ given $\dii{1}$ by $\tilde{\pi}_{ji}=\mathrm{Pr}(\mathbf{s}=s_j\mid|\dii{1}|=|x_i|)$.

Next, we study the receivers listed in Table~\ref{table1}. 
We begin by introducing the conventional MFS-MD {and  MFS-PR} receivers. 
Then, we study  MFS-MAP, which is used to determine the performance limits of  MFS demodulation. 
Next, we devise the optimal receiver, SS-MAP, which serves as a benchmark to assess the performance of the other receivers.
 Finally, two heuristic receivers, MxM-MD and MxM-TS are studied.
  These receivers  have lower complexity than SS-MAP and can obtain arbitrary low SERs for sufficiently high powers.

\subsection{ MFS demodulation with MD detection (MFS-MD)}

The MFS demodulator  maps the received signal $\Aa{1}{L}{t}$ to the complex number
\begin{align}\label{mfs_out}
\omfs&=\int_{0}^{T}\Aa{1}{L}{t}\cdot g(t)\di t\\
&=\langle \Aa{1}{L}{t}, g(t)\rangle.
\end{align}
After observing the demodulation outcome  $\omfs=\omfss$, 
the MD detector selects $x_m\in \mathcal{X}$  such that
\begin{align}
	m
	&=\arg\min_{i}|v-x_i|^2.\label{eq22d}
\end{align}

{
\subsection{MFS demodulation with phase recovery (MFS-PR) }
 In the MFS-PR receiver, the output of the MFS demodulator passes through a phase-recovery block and is then  fed to the MD detector.
Throughout, we shall focus on the phase-recovery technique proposed in \cite{pfau_2009_jlt}.\footnote{ The  test carrier phases considered in the simulation results are $\pi b/128$, $b\in\{-32, \dots , 31 \}$. }}

\subsection{ MFS demodulation with MAP detection (MFS-MAP)}

Given the MFS output $\omfs=v$ in \eqref{mfs_out}, the optimal MAP detector  determines the input symbol $x_m\in \mathcal{X}$, such that 
\begin{align}
	m
	&=\arg\max_{i}\mathrm{Pr}\rlpo{\dii{1}=x_i\mid\omfs=\omfss}\label{eq22}\\
	&=\arg\max_{i}\pi_i f_{\omfs\mid \dii{1}}\rlpo{\omfss\mid x_i}\\
	&=\arg\max_{i}\pi_i\sum_{j}\tilde\pi_{ji}f_{\omfs\mid \mathbf s, \dii{1}}\rlpo{\omfss\mid s_j, x_i}\label{eqq2}
\end{align}
where in  \eqref{eqq2} we used that $\mathrm{Pr}(\mathbf{s}=s_j\mid\dii{1}=x_i)=\mathrm{Pr}(\mathbf{s}=s_j\mid|\dii{1}|=|x_i|)=\tilde{\pi}_{ji}.$ The conditional probability
$f_{\omfs\mid \mathbf s, \dii{1}}\rlpo{\omfss\mid s_j, x_i}$
can be calculated by noting that, given $\mathbf{s}=s_j$ and  $\dii{1}=x_i$, we have that $\omfs\distas\mathcal{CN}\rlpo{\mu_{ji},N_0}$, where 
\begin{equation}
\mu_{ji}=  x_i\left\langle  g(t)\exp\rlpo{j\eta_1s_jg^2(t)} , g(t)\right\rangle.
\end{equation}
Therefore,
\begin{equation}\label{map22}
f_{\omfs\mid \mathbf s, \dii{1}}\rlpo{\omfss\mid s_j, x_i}=\frac{1}{\pi N_0}\exp\rlpo{-\frac{|\omfss-\mu_{ji}|^2}{N_0}}.
\end{equation}
\subsection{ Sufficient statistics  with MAP  detection (SS-MAP)}

Let $\phii{\si}{t}=\eta_1\,{\si}\, g^2(t).$ 
The real  and the imaginary  part of $	\Aa{1}{L}{t}$ are  
\begin{IEEEeqnarray}{rCl}
	\Aar{1}{L}{t}&=& \red\rlpo{\dii{1}}g(t)\cos\rlpo{\phii{\si}{t}}-\ima\rlpo{\dii{1}}g(t)\sin\rlpo{\phii{\si}{t}}+\Nnr{1}{t}\label{17}\\
	\Aai{1}{L}{t}&=& \red\rlpo{\dii{1}}g(t)\sin\rlpo{\phii{\si}{t}}+\ima\rlpo{\dii{1}}g(t)\cos\rlpo{\phii{\si}{t}}+\Nni{1}{t}.
\end{IEEEeqnarray}
Note that, if additive noise is neglected, the signals $\Aar{1}{L}{t}$ and $\Aai{1}{L}{t}$ can be written as  linear combinations of the signals $\hs{\ell}{t}=g(t)\sin\rlpo{\phii{s_\ell}{t}}$ and $\hst{\ell}{t}=g(t)\cos\rlpo{\phii{s_\ell}{t}}$, $\ell=1, \dots, |\mathcal{S}|$.
Therefore, by \cite[Corollary 26.4.2]{lapidoth2009foundation},
\begin{align}
	\mathbf u_\ell^{\mathrm{R}}&=\langle\Aar{1}{L}{t},\hs{\ell}{t}\rangle\label{ss1}\\
	\tilde{\mathbf u}_\ell^{\mathrm{R}}&=\langle\Aar{1}{L}{t},\hst{\ell}{t}\rangle\\
	\mathbf u_\ell^{\mathrm{I}}&=\langle\Aai{1}{L}{t},\hs{\ell}{t}\rangle\\
	\tilde{\mathbf u}_\ell^{\mathrm{I}}&=\langle\Aai{1}{L}{t},\hst{\ell}{t}\rangle\label{ss4}
\end{align}
are sufficient statistics for determining $\dii{1}$ based on $\Aa{1}{L}{t}$. 
Let  $\ul{\mathbf u}^{\mathrm{R}}=[\mathbf u_1^{\mathrm{R}}, \dots, \mathbf u_{|\mathcal{S}|}^{\mathrm{R}}]$, and similarly define the vectors $\tilde{\ul{\mathbf u}}^{\mathrm{R}}$, $\ul{\mathbf u}^{\mathrm{I}}$, and $\tilde{\ul{\mathbf u}}^{\mathrm{I}}$.
Moreover, let the vector $\ul {\mathbf u}$  with length $4|\mathcal{S}|$ be the concatenation of the aforementioned vectors, i.e.,
\begin{equation}\label{eq:u}
\ul {\mathbf u}=[\ul{\mathbf u}^{\mathrm{R}}, \tilde{\ul{\mathbf u}}^{\mathrm{R}}, \ul{\mathbf u}^{\mathrm{I}}, \tilde{\ul{\mathbf u}}^{\mathrm{I}}].
\end{equation} 
It follows from \cite[Prop.~25.15.2]{lapidoth2009foundation} that the vector $\vU$  is conditionally jointly Gaussian given  $\si=s_j$ and $\dii{1}=x_i$.
Let the conditional mean vector of $\vU$ given $\mathbf{s}=s_j$ and $\dii{1}=x_i$ be $\ul{\mu}_{ji}$ and the conditional covariance matrix be ${\Sigma}$ (as we shall see later, ${\Sigma}$ does not depend on $j$ or $i$). It follows from steps similar to \eqref{eq22}--\eqref{eqq2} that the MAP decoder, after observing $\vU=\vUs$, selects the transmitted symbol  $x_m$ such that 
\begin{align}
	m
	&=\arg\max_{i}\pi_i\sum_{j}\tilde\pi_{ji}f_{\vU\mid \mathbf s, \dii{1}}\rlpo{\ul{u}\mid s_j, x_i}\label{eqq}
\end{align}
where
\begin{align}\label{eq25}
	f_{\vU\mid \mathbf s, \dii{1}}\rlpo{\ul{u}\mid s_j, x_i} =\frac{\exp\rlpo{-\frac{1}{2}(\vUs-\ul{\mu}_{ji})\Sigma^{-1}(\vUs-\ul{\mu}_{ji})^{T}}}{(2\pi)^{2|\mathcal S|}\sqrt{|\Sigma|}}.
\end{align}

Next we calculate  $\ul{\mu}_{ji}$ and $\Sigma$. We write $\ul{\mu}_{ji}$ as a concatenation of four vectors: $\ul{\mu}_{ji}=[\ul{\mu}_{ji}^\mathrm{R},\, \ul{\tilde{\mu}}_{ji}^\mathrm{R},\, \ul{\mu}_{ji}^\mathrm{I},\, \ul{\tilde{\mu}}_{ji}^\mathrm{I}]$. It follows from \eqref{eq:u} that the $\ell$th element of $\ul{\mu}_{ji}^\mathrm{R}$ is
\begin{equation}\label{md}
{\mu}_{ji\ell}^\mathrm{R}=\expe\rlbo{\mathbf u_\ell^{\mathrm{R}}\mid \mathbf{s}=s_j , \dii{1}=x_i }.
\end{equation}
The vectors $\ul{\tilde{\mu}}_{ji}^\mathrm{R}$, $\ul{\mu}_{ji}^\mathrm{I}$, and  $\ul{\tilde{\mu}}_{ji}^\mathrm{I}$ can be calculated  as in \eqref{md}.
We have from \eqref{17} and \eqref{ss1} that
\begin{align}
	\expe\rlbo{\mathbf u_\ell^{\mathrm{R}}\mid \mathbf{s}=s_j , \dii{1}=x_i }&
	=\red\rlpo{\diin{i}}
	\left\langle g(t)\cos\rlpo{\phii{s_j}{t}} , \hs{\ell}{t}\right\rangle-\ima\rlpo{\diin{i}}\left\langle g(t)\sin\rlpo{\phii{s_j}{t}} , \hs{\ell}{t}\right\rangle.
\end{align}
Moreover,
\begin{align}
	&\left\langle g(t)\cos\rlpo{\phii{s_j}{t}} , \hs{\ell}{t}\right\rangle\\
	&=\int_{0}^{T}g^2(t)\sin\rlpo{\phii{s_\ell}{t}}\cos\rlpo{\phii{s_j}{t}}\di t\\
	&=\frac{1}{2}\int_{0}^{T}g^2(t)\sin\rlpo{\phii{s_\ell+s_j}{t}}\di t +\frac{1}{2}\int_{0}^{T}g^2(t)\sin\rlpo{\phii{s_\ell-s_j}{t}}\di t\\
	&=\FFIIs{s_\ell+s_j}+\FFIIs{s_\ell-s_j}
\end{align}
where we have set
\begin{equation}
\FFIIs{z}=\frac{1}{2}\int_{0}^{T}g^2(t)\sin\rlpo{\phii{z}{t}}\di t.
\end{equation}
Similarly,
\begin{align}
	\left\langle g(t)\sin\rlpo{\phii{s_j}{t}} , \hs{\ell}{t} \right\rangle&
	=\FFIIc{s_\ell-s_j}-\FFIIc{s_\ell+s_j}
\end{align}
where
\begin{equation}
\FFIIc{z}=\frac{1}{2}\int_{0}^{T}g^2(t)\cos\rlpo{\phii{z}{t}}\di t.
\end{equation}
Therefore,
\begin{align}
	{\mu}_{ji\ell}^\mathrm{R}&=\red\rlpo{\diin{i}}
	\FFIIs{s_\ell+s_j}+\red\rlpo{\diin{i}}\FFIIs{s_\ell-s_j}+\ima\rlpo{\diin{i}}\FFIIc{s_\ell+s_j}
	-\ima\rlpo{\diin{i}}\FFIIc{s_\ell-s_j}.\label{uji:begin}
\end{align}
With analogous calculations, we obtain
\begin{align}
	{\tilde \mu}_{ji\ell}^\mathrm{R}&=	\red\rlpo{\diin{i}}
	\FFIIc{s_\ell+s_j}+\red\rlpo{\diin{i}}\FFIIc{s_\ell-s_j}-\ima\rlpo{\diin{i}}\FFIIs{s_\ell+s_j}+\ima\rlpo{\diin{i}}\FFIIs{s_\ell-s_j}\\
	{ \mu}_{ji\ell}^\mathrm{I}&=-\red\rlpo{\diin{i}}\FFIIc{s_\ell+s_j}+\red\rlpo{\diin{i}}\FFIIc{s_\ell-s_j}+
	\ima\rlpo{\diin{i}}\FFIIs{s_\ell+s_j}+\ima\rlpo{\diin{i}}\FFIIs{s_\ell-s_j}\\
	{ \tilde \mu}_{ji\ell}^\mathrm{I}&=\red\rlpo{\diin{i}}\FFIIs{s_\ell+s_j}-\red\rlpo{\diin{i}}\FFIIs{s_\ell-s_j}+
	\ima\rlpo{\diin{i}}\FFIIc{s_\ell+s_j}+\ima\rlpo{\diin{i}}\FFIIc{s_\ell-s_j}.\label{uji:end}
\end{align}

Next, we calculate $\Sigma$, which is a $4|\mathcal{S}|\times 4|\mathcal{S}|$ matrix. Dividing $\Sigma$ into 16 submatrices of size $|\mathcal{S}|\times|\mathcal{S}|$ and using  \cite[Prop.~25.15.2]{lapidoth2009foundation}, we obtain 
\begin{equation}\label{sig1}
\Sigma=\begin{bmatrix}
\Sigma^{11}       & \Sigma^{12} &0 & 0 \\
\rlp{\Sigma^{12}}^T       & \Sigma^{22} &0 & 0 \\
0       & 0 &\Sigma^{11} & \Sigma^{12} \\
0      & 0 &\rlp{\Sigma^{12}}^T & \Sigma^{22} \\
\end{bmatrix}
\end{equation}
where the element $\Sigma_{k\ell}^{11}$ of the submatrix $\Sigma^{11}$ is
\begin{align}
	\Sigma_{k\ell}^{11}&=\frac{N_0}{2}\left\langle\hs{k}{t},\hs{\ell}{t}\right\rangle \\
	&=\frac{N_0}{2}\int_{0}^{T}g^2(t)\sin\rlpo{\phii{s_k}{t}}\sin\rlpo{\phii{s_\ell}{t}}\di t\\
	&=\frac{N_0}{4}\int_{0}^{T}g^2(t)\cos\rlpo{\phii{s_k-s_\ell}{t}}\di t -\frac{N_0}{4}\int_{0}^{T}g^2(t)\cos\rlpo{\phii{s_\ell+s_k}{t}}\di t \\
	&=\frac{N_0}{2}\rlb{\FFIIc{s_k-s_\ell}-\FFIIc{s_k+s_\ell}}
\end{align}
for $k=1, \dots, |\mathcal S|$ and $\ell=1, \dots, |\mathcal S|$.
Furthermore,
\begin{align}
	\Sigma_{k\ell}^{12}&=\frac{N_0}{2}\left\langle\hs{k}{t},\hst{\ell}{t}\right\rangle
	=\frac{N_0}{2}\rlb{\FFIIs{s_k+s_\ell}+\FFIIs{s_k-s_\ell}}\\
	\Sigma_{k\ell}^{22}&=\frac{N_0}{2}\left\langle\hst{k}{t},\hst{\ell}{t}\right\rangle
	=\frac{N_0}{2}\rlb{\FFIIc{s_k+s_\ell}+\FFIIc{s_k-s_\ell}}.\label{sig22}
\end{align}
{ Note that the real and the imaginary parts of $\Nn{1}{t}$ are independent processes. This explains why half of the elements in \eqref{sig1} are zero. }

\subsection{ MxM  demodulation with MD detection (MxM-MD)}

Next, we present a novel heuristic demodulation  scheme, which is composed of three steps. First, the phase distortion of the received signal is estimated. This phase distortion is compensated for in the second step. Third, a MFS is applied to obtain the output of the demodulator.
The first step is based on the following proposition, whose proof follows from \cite[Ch.~4, Eq.~(3)]{ahlfors1953complex}.
\begin{proposition}\label{p1}\normalfont
 Let $f(t)$ be a nonnegative continuous  function on the interval $[a , b]$. Then 
\begin{equation}
\max_{s\in \mathbb R} \left|\int_{a}^{b}f(t)e^{jsf(t)}\di t\right|=\int_{a}^{b}f(t)\di t
\end{equation}
and $s=0$ achieves the maximum.
\end{proposition}

Next, we use Proposition~\ref{p1} to devise the first step of the demodulation. Assume that  $\dii{1}=x_1$ and $\mathbf{s}=s$. 
To estimate $s$, the receiver  calculates
\begin{IEEEeqnarray}{rCl}\label{smax}
	s_{\text{max}} &=& \arg\!\max\limits_{s'\in \mathcal{S}}\left|\int\limits_{0}^{T} \Aa{1}{L}{t}\cdot g(t)e^{-j\eta_1 s' \, g^2(t)}\di t\right|\\
	&=&\arg\!\max\limits_{s'\in \mathcal{S}}\left| x_1\int\limits_{0}^{T} g^2(t)e^{j\eta_1 ( s-s') \, g^2(t)}\di t +\mf n\right|\label{eq45}
\end{IEEEeqnarray}
where $\mf n \distas \mathcal{CN}(0,N_0)$. If we ignore the noise in \eqref{eq45}, it follows from Proposition~1 that $s_{\text{max}}=s$.
Therefore,  $s_\text{max}$ calculated in \eqref{smax} provides an estimate of $s$ in the  presence of  noise.  Note that, similar to the SS decoder, the computation of  $s_\text{max}$ in \eqref{eq45} requires $4|\mathcal{S}|$ real-valued correlators.

In the next step, the phase distortion is compensated for by multiplying the received signal with $\exp\rlpo{-j\eta_1 s_\text{max} \, g^2(t)}$. Finally, the result is fed to the MFS demodulator. To summarize, the output of the MxM demodulator is
 	\begin{equation}
\omtd=\int\limits_{0}^{T} \Aa{1}{L}{t}\cdot g(t)e^{-j\eta_1 s_{\text{max}}\, g^2(t)}\di t. \label{eq:mxm}
 \end{equation}
We see from \eqref{ch} that if the demodulator successfully compensates for the phase distortion, i.e., if $s_\text{max}=s$, then the output of the MxM demodulator has a Gaussian distribution centered at $x_1$ with variance $N_0$.
However, if $s$ is not estimated correctly at the receiver, the output of the demodulator has a different mean. 
The MD detector determines  $x_m\in \mathcal{X}$, based on the MxM output $\omtd=w$, such that
\begin{align}
	m
	&=\arg\min_{i}|w-x_i|^2.\label{eq22dd}
\end{align}

\subsection{ MxM demodulation with TS detection (MxM-TS)}

 To map the output of the MxM demodulator $\omtd=w$ in \eqref{eq:mxm} to one of the constellation points, MxM-TS uses a  simple two-stage detector. 
The two-stage detector first estimates the amplitude of the transmitted signal and then determines its phase.
Specifically, let $\mathcal{R}=\left\{r_1, \dots, r_{|\mathcal{R}|}\right\}$ be the set of all  possible amplitudes of the transmitted symbol. The amplitude detector chooses  $\hat R=r_i, \  1\leq i\leq |\mathcal{R}|,$ if 
$m_{i-1}\leq|w|\leq m_i$, where $m_i$ is  the $i$th detection threshold.
 To compute the thresholds $m_i$, we  assume that   given $\dii{1}=x_1$, we have $\omtd\distas\mathcal{CN}(x_1,N_0)$. Therefore,
\begin{align}\label{de}
	f_{|\omtd|\, {\bm{\vert}}\,|\dii{1}|} \rlpo{m_i\mid r_i}&\approx \frac{2m_i}{N_0}\exp\rlpo{-\frac{{m_i^2+r_i^2}}{N_0}}I_0\rlpo{\frac{2m_ir_i}{N_0}}
\end{align}
where $I_0(\cdot)$ is the zeroth order modified Bessel function of the first kind. Since the approximated conditional distribution in \eqref{de} is  unimodal,    $m_i$ can be  obtained based on the MAP rule by solving 
\begin{align}
	\mathrm{Pr}\rlpo{|\dii{1}|=r_i}f_{|\omtd|\,\bm|\,|\dii{1}|} \rlpo{m_i\mid r_i}=\mathrm{Pr}\rlpo{|\dii{1}|=r_{i+1}}f_{|\omtd|\,|\,|\dii{1}|}\rlpo{m_{i}\mid r_{i+1}}\label{m_i}
\end{align}
 for $i=1, \dots, |\mathcal R|-1$, (with the convention that $m_0=0$ and $m_{|\mathcal{R}|}=\infty$).
After estimating the amplitude of the transmitted signal, the two-stage detector selects the constellation point with amplitude $\hat R$ that is closest to $w$.

\subsection{Complexity} The MFS demodulator calculates the correlation between a real function $g(t)$ and  the complex received signal, which can be implemented by two real-valued correlators (or, equivalently, filters). This number is $4 |\mathcal{S}|$ for the more sophisticated  SS and MxM demodulators. The MAP detector in \eqref{eqq}--\eqref{eq25}  involves calculating a quadratic form in a $4|\mathcal S|$--dimensional space, which makes it much more computationally demanding than the other detectors.  The MAP detector in \eqref{eqq2}--\eqref{map22} involves calculating $|\mathcal S|$ exponential functions. Hence, it is  more complex than the MD and TS detectors, which are only based on comparisons. { Moreover,   the MxM and the SS demodulators have  larger bandwidths than the MFS demodulator. The  bandwidth of the MxM demodulator is the maximum of the bandwidths of the signals $g(t)\exp(-j\eta s g^2(t))$ over all values of $s$;   the bandwidth of the SS demodulator  is the maximum of the bandwidths of the signals $\hs{\ell}{t}$ and $\hst{\ell}{t}$ over $\ell=1, \dots, |\mathcal{S}|$.
}

 {
 \section{Asymptotic  SER Analysis}\label{siv}
In this section, we provide analytical evaluations of the asymptotic SER of the proposed receivers.
 Let the input random variable be $\dii{1}=\sqrt{\mathcal P}\diip{1}$, where $\diip{1}$ takes values from a fixed alphabet set $\mathcal{X}'=\{x'_1, x'_2,..., x'_{|\mathcal{X}'|} \}$, with some arbitrary probability distribution.  Similarly, let $\dii{2}=\sqrt{\mathcal P}\diip{2}$.
 In order to make analytical calculations possible, we  assume triangular pulse shaping, i.e.,
 \begin{equation}\label{tps}
g(t)= c\rlpo{\frac{T}{2}-\left|\frac{T}{2}-t\right|},   \ \ \ \ \ \ c=\sqrt{\frac{12}{T^3}}. 
 \end{equation}
 The following theorem presents our asymptotic SER results.
  \begin{theorem}\label{thm1}\normalfont
Assuming triangular pulse shaping, 
\begin{enumerate}[i)]
	\item  the  SER of the MFS-MAP receiver goes to $1-\max_{i}(\pi_i)$ as $\mathcal P\to\infty$, where ${\pi_i=\mathrm{Pr}(\dii{1}=x_i)}$.  
	\item the SER of the MxM-MD and MxM-TS receivers goes to zero as  $\mathcal P\to\infty$.
\end{enumerate}
  \end{theorem}

\begin{proof}\normalfont
Substituting     \eqref{ch}   into \eqref{mfs_out}, we can write the output of the MFS demodulator as
\begin{align}\label{mfs_out2}
\omfs&=\int_{0}^{T}\dii{1}g^2(t)\exp\rlpo{j\eta_1\rlp{|\dii{1}{}|^2+2|\dii{2}{}|^2}g^2(t)}\di t + \mathbf n\\
&=2c^2\diip{1}\int_{0}^{T/2}\sqrt{\mathcal P}  t^2\exp\rlpo{j\eta_1c^2\mathcal P\rlp{|\diip{1}{}|^2+2|\diip{2}{}|^2}t^2}\di t + \mathbf n.\label{eq2}
\end{align}
where $\mathbf n\distas\mathcal{CN}\rlpo{0,N_0}$. Here, \eqref{eq2} follows from \eqref{tps} and the definitions of $\mf x'_1$ and $\mf x'_2$.  We first assume $\mf x'_1\neq0 $. It can be shown by standard algebraic calculations that the integral in \eqref{eq2} is $\bo\rlpo{1/\sqrt{\mathcal P}}$. Therefore, as $\mathcal P\to \infty$ the first term in \eqref{eq2} goes to zero.  Furthermore, this term is zero if  $\mf x'_1=0$. Since $\mathbf n$ is independent of the transmitted signal,  the MAP detector selects, in the limit $\mathcal P\to\infty$, the symbol with largest a priory probability regardless of received signal, resulting in a SER of $1-\max_{i}(\pi_i)$.

Next, we prove the second part of the theorem.  Focusing on \eqref{eq45}, one can show with  similar calculations as above that for every $s\neq s'$, the integral $x_1\int g^2(t)e^{j\eta_1 ( s-s') \, g^2(t)}\di t$ goes to zero as $\mathcal P\to \infty$. Moreover, for $s= s'$, the integral equals  $x_1$. Therefore, assuming $ x_1\neq0 $, we conclude  that, in the limit $\mathcal P\to \infty$,  we have  $s_{\text{max}}=s$ with probability one. Under the assumption that $s_{\text{max}}=s$, it follows  from \eqref{eq:mxm} that    $\omtd\distas\mathcal{CN}\rlpo{\sqrt{\mathcal{P}}x'_1,N_0}$, where $\omtd$ is the outcome of the MxM demodulator. Therefore, in the limit  $\mathcal P\to \infty$, both MD and TS detectors will correctly detect the symbol $x_1$  with probability one.  If $ x_1=0 $, then $\omtd$ does not depend on $\mathcal P$, and therefore in the limit $\mathcal P\to \infty$, the symbol $0$ will be correctly detected  by both MD and TS with probability one.   
\end{proof}

   	Note that the first result in Theorem~\ref{thm1} implies that the asymptotic SER of the MFS-MD and MFS-PR is lower-bounded by $1-\max_{i}(\pi_i)$; the second result implies that the SER of SS-MAP goes to zero as $\mathcal P\to\infty$. 

%
 
}

\section{Numerical Examples }\label{s4}

{ This section presents numerical SER evaluations   for three single-span channels. The simplified model \eqref{cont_1}--\eqref{cont_2} is studied in Section~\ref{s41} and two  NLS channels are analyzed via split-step Fourier simulations in  Section~\ref{s5.5}.
 }

{\subsection{Transmission Over the Simplified Channel \eqref{cont_1}--\eqref{cont_2}}\label{s41}
}
In this section, we evaluate the performance of the {six} receivers presented in the previous section, by conducting Monte Carlo simulations on the channel model \eqref{cont_1}--\eqref{cont_2}.
We consider the transmission of 16-ary quadrature amplitude modulation (QAM) data symbols from each of the two transmitters.
 The input power $\mathcal P=E_s/T$, where $E_s=\expe[|\dii{1}|^2]$, is assumed to be the same for both channels.
 For these choices we have that  $|\mathcal{S}|=7$ and
 \begin{equation}
\mathcal{S}=\left\{0.6E_s, \     1.4E_s,\    2.2E_s,\    3E_s,\    3.8E_s,\    4.6E_s,\    5.4E_s\right\}.
 \end{equation}
The simulation parameters can be found in Table~\ref{t2}.
The nonlinear coefficient can be calculated from \eqref{eta}--\eqref{leff} as {$\eta_1=\eta_2= 22.1 \ \mathrm{W}^{-1}$}. 
Also, using  \eqref{eq:noise:var}, one  obtains {${N_0=1.43\cdot 10^{-15} \ \mathrm{W}/\mathrm{Hz}}$}.  We use ${100}$ samples per symbol and set $g(t)$ to a truncated Gaussian pulse  with a full width at half maximum of $T/2$. 
 A  uniform input distribution is assumed for both   transmitted signals, i.e., $\pi_i=1/16$. Consequently, the conditional probabilities $\tilde \pi_{ji}$ can be calculated as in Table~\ref{T1}.

\begin{table}[!t]
	\renewcommand{\arraystretch}{1.3}
	\caption{Parameters used in the simulation.}
	\label{t2}
	\centering
	\begin{tabular}{c c c }
		\hline
		\hline
		Parameter&Symbol& Value\\
		\hline
		Span length & $L_{\text{span}}$&$150 \ \mathrm{km}$\\
		Attenuation &$\alpha$ & $0.25 \ \mathrm{dB/km}$\\
		Nonlinearity& $\gamma_1=\gamma_2$ & $1.27 \ \rlp{\mathrm{W km}}^{-1}$\\
		Symbol rate & $1/T$&$10 \ \mathrm{Gbaud}$\\
		Optical photon energy & $h\nu$&$1.28\cdot 10^{-19} \ \mathrm{J}$\\
		Amplifier noise figure & $F$&$6 \ \mathrm{dB}$\\
		 Number of spans  & $n_{\text{span}}$&$1$\\
		\hline
		\hline
	\end{tabular}
	\vspace{-0.3cm}
\end{table}

Fig.~\ref{Fig1} shows  scatter plots of the MFS  and MxM demodulator outputs for three levels of input power. 
Note that since the output of the SS demodulator \eqref{ss1}--\eqref{ss4} lies in a vector space with dimension $4|\mathcal S|=28$, it is not possible to draw its scattering pattern.
 At { $\mathcal P =-5$ }dBm, it can be seen from Fig.~\ref{Fig1}\subref{fig1a} that the output of the MFS demodulation follows approximately  a Gaussian distribution. 
 However, the  clouds are not centered at the constellation points. Rather, they are rotated by an amount  proportional to the amplitude square of the constellation points. This rotation is caused by the SPM distortion.
In Fig.~\ref{Fig1}\subref{fig1d}, the output of the MxM demodulator at { $\mathcal P =-5$} dBm is shown.  It can be seen that with this demodulator, the effect of SPM is mitigated. Indeed,  the clouds are now centered at the constellation points.  

\begin{table}[!t]
	\centering
	\caption{$\tilde{\pi}_{ji}=\mathrm{Pr}(\mathbf{s}=s_j\mid|\mathbf{x}|=|x_i|)$ for $16$-QAM transmission with uniform distribution. }
	\label{T1}
	\begin{tabular}{|l||*{8}{c|}}\hline
		\backslashbox{$|x_i|$}{$s_j$}
		&$0.6E_s$&$1.4E_s$&$2.2E_s$&$3E_s$&$3.8E_s$&$4.6E_s$&$5.4E_s$\\\hline\hline
		$\sqrt{0.2E_s}$ &$0.25$&$0$&$0.5$&$0$&$0.25$&$0$& $0$\\\hline
		$\sqrt{E_s}$ &$0$&$0.25$&$0$&$0.5$&$0$&$0.25$& $0$\\\hline
		$\sqrt{1.8E_s}$&$0$&$0$&$0.25$&$0$&$0.5$&$0$& $0.25$\\\hline
	\end{tabular}
\end{table}

	\begin{figure*}[!t]
	\centering
	\subfloat[]{\includegraphics[scale=.8]{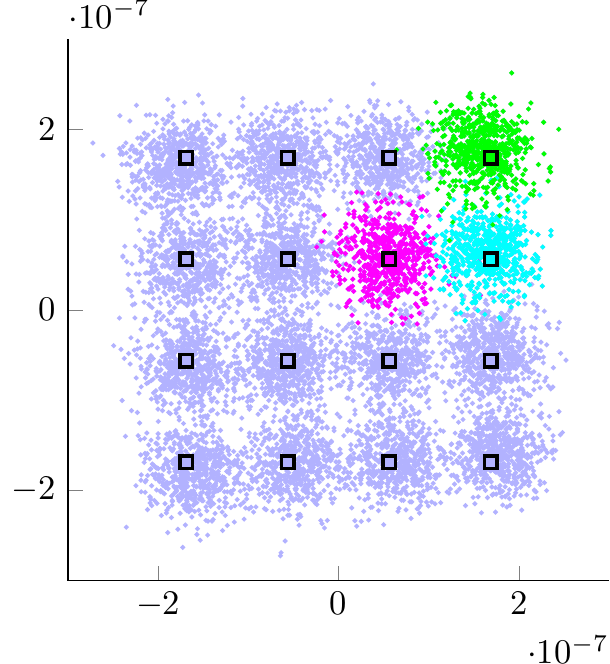}%
		\label{fig1a}}
	\hfil
	\subfloat[]{\includegraphics[scale=.8]{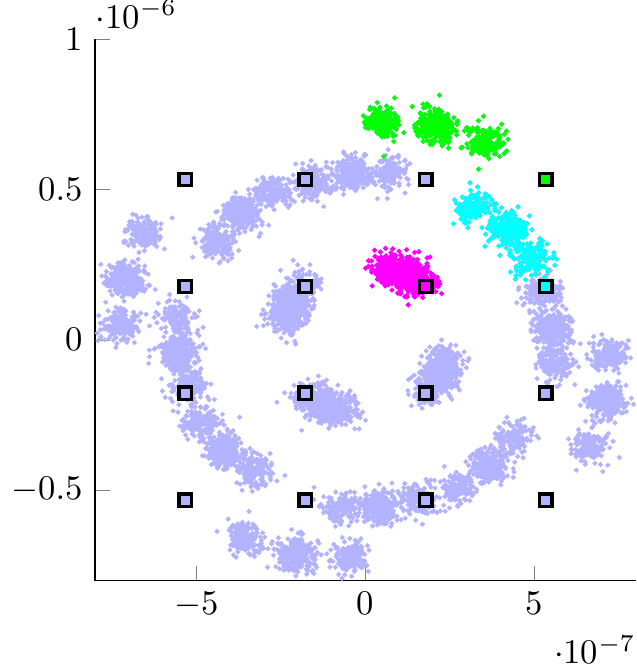}%
		\label{fig1b}}
	\subfloat[]{\includegraphics[scale=.8]{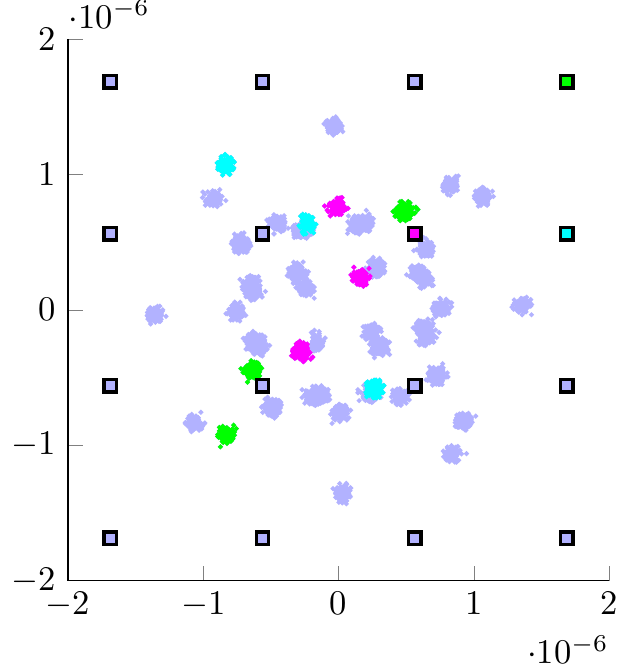}%
		\label{fig1c}}
	\\
	\subfloat[]{\includegraphics[scale=.8]{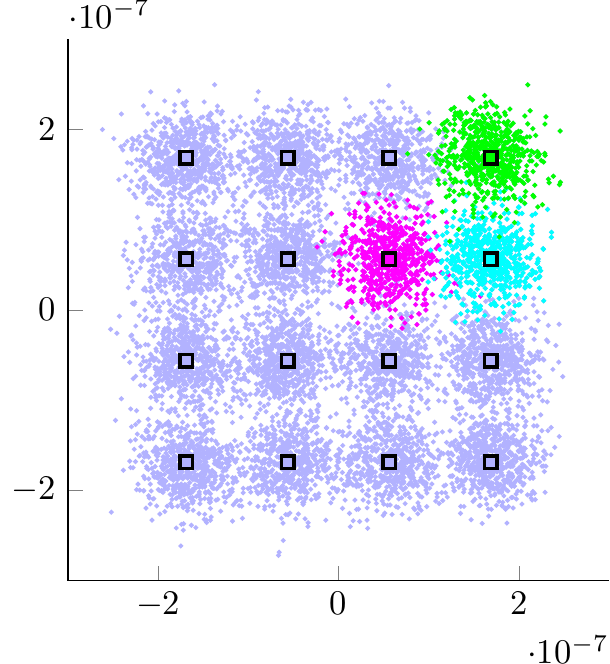}%
		\label{fig1d}}
	\subfloat[]{\includegraphics[scale=.8]{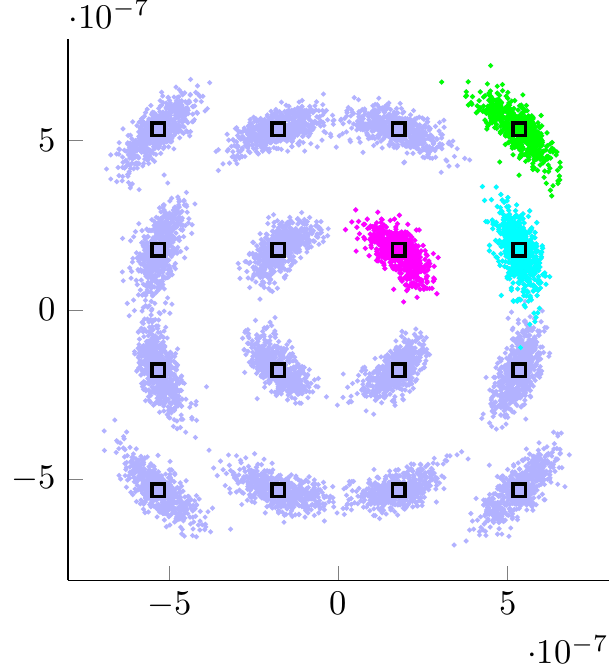}%
		\label{fig1e}}
	\subfloat[]{\includegraphics[scale=.8]{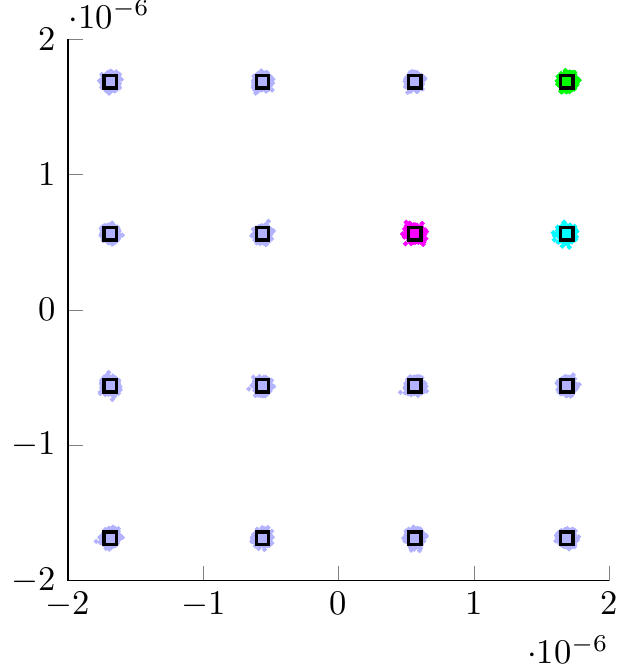}%
		\label{fig1f}}
	\caption{\footnotesize  Scatter plots of the output of two demodulation schemes, matched filtering and sampling (MFS) and  maximum matching (MxM), for $16$-QAM with three input powers {$\mathcal P =-5$} dBm: (a) and (d), { $\mathcal P =5$} dBm: (b) and (e), and  {$\mathcal P =15$} dBm: (c) and (f). Different colors are used to identify demodulator outputs corresponding to three given input symbols.
	}
	\label{Fig1}
	\vspace{-.3cm}
\end{figure*}

Fig.~\ref{Fig1}\subref{fig1b} illustrates the MFS demodulator's output at { $\mathcal P =5$} dBm. 
One can observe that the effect of the nonlinear distortion becomes more significant compared to the case { $\mathcal P =-5$} dBm.
Each constellation point is scattered to three different clouds, each one corresponding to the three  possible values of the XPM distortion (the three values of $|\dii{2}|$). Also, the centers of the clouds are further rotated away from the constellation points because of the SPM.  
As shown in Fig.~\ref{Fig1}\subref{fig1e}, the output of MxM is also dispersed to three clouds per  symbol. However, unlike MFS, these clouds are centered at the constellation points.

	One can observe from Fig.~\ref{Fig1}\subref{fig1c}  that when { $\mathcal P =15$} dBm both  the phase and the amplitude of the MFS output are distorted. The power loss, which is evident in   Fig.~\ref{Fig1}\subref{fig1c}, can be explained  as follows. At high powers, the phase of the integrand in \eqref{mfs_out} changes quickly during one time slot. 
	This rapid phase change scales down the integral's result in \eqref{mfs_out}, which is the output of the MFS demodulator.
	 Alternatively, the power loss can be explained in the frequency domain.  
	 At high powers, the nonlinear distortions substantially broaden the signal's spectrum. However, MFS uses a filter matched to the transmitted pulse shape, which has the same bandwidth as the transmitted signal. Therefore, the signal's out-of-band energy is excluded. 
	  It can be seen from Fig.~\ref{Fig1}\subref{fig1f} that the output of the MxM demodulator is centered at each constellation point, i.e., there is no power loss or phase distortion.  Fig.~\ref{Fig1}\subref{fig1f} indicates that the nonlinear distortion is effectively compensated for by the MxM demodulator.
	
	Fig.~\ref{err} depicts the SER for the { six} receivers  introduced in Section~\ref{s3}. Moreover, the SER  for the AWGN channel, obtained by setting $\eta_1=0$ in \eqref{ch}, is plotted for comparison. In the following, we discuss the results in Fig.~\ref{err} for each demodulation scheme. 
	
	\textit{MFS demodulator:} In our analysis, this demodulator is combined with two detectors, namely, MD and MAP. It is well known that for the AWGN channel and a uniform input distribution, these two detectors coincide. 
	On the contrary, it can be observed in Fig.~\ref{err} that for the nonlinear channel considered here, a substantial gap exists between the performance of these two detectors. 
	The SER for the MFS-MD receiver follows first the SER of the AWGN channel, reaches a minimum point of {$1.6\cdot 10^{-2}$}, and then increases to approximately   one at high power levels. The increase in the SER in the high-power regime can be explained by looking at Figs.~\ref{Fig1}\subref{fig1a}--\subref{fig1c}. The output of the MFS demodulator is not centered at the constellation points. Therefore, the MD decoder fails to provide a sound estimate of the transmitted symbols.  { Comparing MFS-PR with MFS-MD, it can be seen that a considerable  improvement  is obtained by performing phase recovery. The minimum SER for MFS-PR is $1.4\cdot 10^{-3}$. }

	By changing the detection scheme from MD to MAP, a substantial performance gain can be obtained. The MFS-MAP receiver yields  a SER of {$3.1\cdot 10^{-4}$ at $\mathcal P =2$ dBm, which is more than $50$} times smaller than the minimum SER that can be obtained with the MFS-MD. 
	The MAP detector can identify the transmitted symbols as long as the output of the MFS consists of well-separated clouds.  
	However, as shown in Figs.~\ref{Fig1}\subref{fig1a}--\subref{fig1c}, because of the nonlinearity, the clouds move in the constellation plane as the power level changes and can overlap. Therefore, based on the position of the clouds, increasing the input power can enhance or deteriorate the performance, which causes the somewhat irregular behavior  of the SER for the MFS-MAP receiver in Fig.~\ref{err}.
	\begin{figure*}[!t]
		\centering
		{\includegraphics{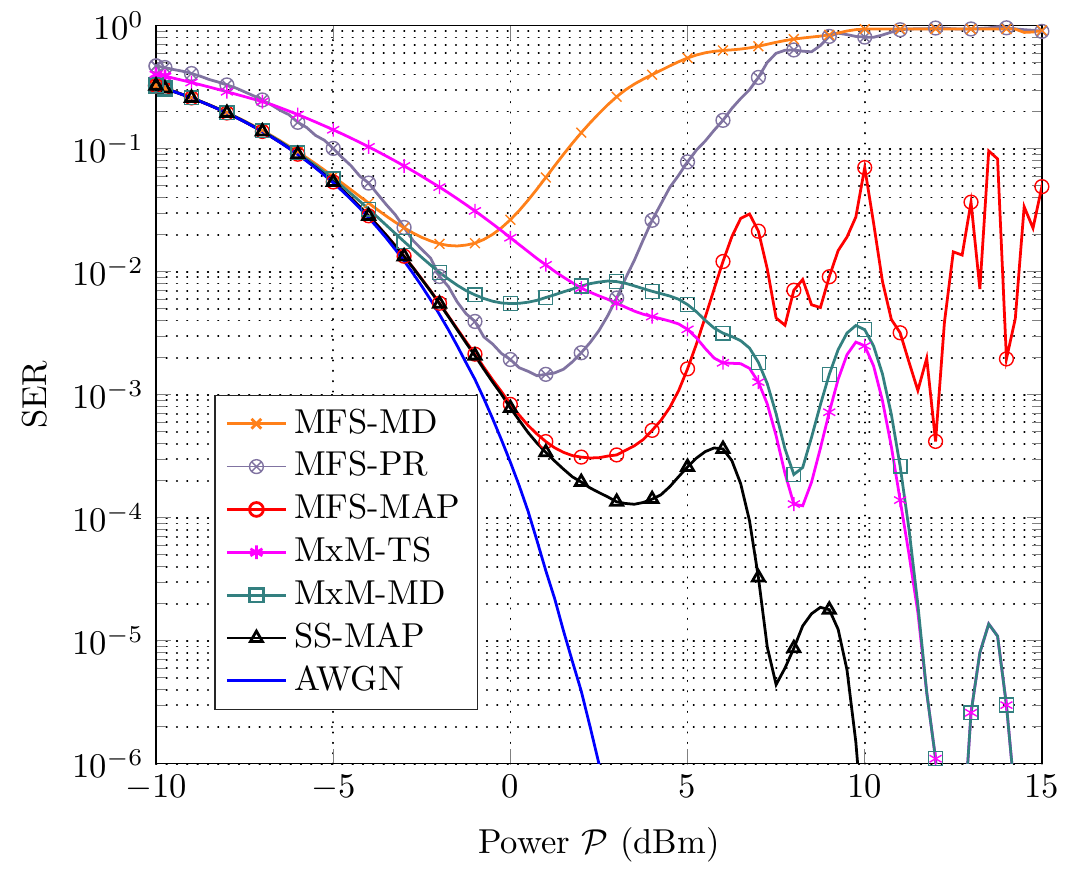}%
		}\hfil
		\caption{\small  The SER of  the {six} receivers introduced in Section~\ref{s3} { is illustrated by conducting Monte-Carlo simulations on the channel model \eqref{cont_1}--\eqref{cont_2}}. The SER of an AWGN channel with the same noise variance is also plotted for comparison.
		}\label{err}

	\end{figure*}
	
		\textit{MxM demodulator:} 
		Two detector schemes, namely, TS and MD, are combined with the MxM demodulator.
		It can be seen in Fig.~\ref{err} that at power levels lower than { $2$}  dBm,   MD  outperforms  TS; when {$2$} dBm $\leq \mathcal P\leq {11}$ dBm,  TS  yields a smaller SER than  MD; and at power levels larger than $11$ dBm, both detectors  perform  equally. The reason is as follows. In the low-power regime, the nonlinearity is weak and the output of the MxM demodulator has approximately a Gaussian distribution centered at the transmitted signal (see Fig.~\ref{Fig1}\subref{fig1d}). Therefore, MD detection is close to optimal at low powers.  
		In the moderate-power regime, the output of the MxM demodulator experiences a phase distortion caused by SPM and XPM  (see Fig.~\ref{Fig1}\subref{fig1e}). 
		In the presence of  phase distortion, TS  outperforms  MD, as previously reported in the literature (see \cite{hager2013design}, for example).
		 Next, we explain why the MxM-MD and the MxM-TS receivers yield the same SER at high powers.
		 The MxM demodulator first tries to cancel the nonlinear distortion. 
		 If it succeeds, the output of the demodulator follows a Gaussian distribution centered at the transmitted symbol.  
		 Otherwise, the outcome of  MxM  gets distorted by the nonlinearity.
		 In the first case, both the MD and TS detectors are able to detect the transmitted symbol almost without error.
		 In the second case, both detectors  make most likely an error because the phase and the amplitude of the demodulator output are severely distorted at high powers.
		 This also causes the nonmonotonic behaviour of the SER as a function of the power. 

			\textit{SS-MAP receiver:} 
			SS-MAP is the optimal receiver for the channel under study (although it has a high complexity) and its SER can serve as a benchmark to compare the performance of other low-complexity receivers.  One can see that the SER of the MFS-MAP follows that of the optimal receiver closely up until {$\mathcal P =1$ }dBm. However, unlike the MFS-MAP, the SER of the SS-MAP and of both the MxM receivers vanishes at high power levels. We see from Fig.~\ref{err} that the effect of the nonlinearity cannot be completely mitigated even by using optimal demodulation and detection schemes, as there exists a considerable gap between the SER of the SS-MAP receiver and SER achievable over an AWGN channel. 
			{
			\subsection{Transmission Over Two Single-Span NLS Channels}\label{s5.5}
			In this section, we evaluate the performance of the  receivers introduced in Section~\ref{s3} for two realistic single-span fiber-optical  systems, one with a low-dispersion SMF and the other with a standard SMF. We use the MFS-PR receiver as a benchmark. We note that the SS demodulator and the MAP detector  no longer represent the optimal demodulation and detection schemes, as they have been designed for the simplified channel model \eqref{cont_1}--\eqref{cont_2} and are mismatched to the channel under study in this section. 
			
			The signals $\Aa{1}{0}{t}$ and $\Aa{2}{0}{t}$ are passed through a brick-wall filter with bandwidth $\df/2$, where $\df$ is the channel spacing parameter in hertz. The baseband input signal, $\Ab{0}{t}$, is generated according to 
			\begin{equation}
				\Ab{0}{t}=\Aa{1}{0}{t}e^{-j\pi t\df}+\Aa{2}{0}{t}e^{j\pi t\df}
			\end{equation}
			  The input signal $\Ab{0}{t}$ is transmitted through the fiber-optical channel governed by the NLS equation 
			\begin{align}
				\frac{\partial\Abd}{\partial z}+\frac{j\beta_{2}}{2}\frac{\partial \Abd}{\partial t^2}+\frac{\alpha}{2}\Abd
				=j\gamma|\Abd|^2\label{ch:nlse}
			\end{align}
			where $\beta_{2}$, $\gamma$, and $\alpha$ are  dispersion, nonlinearity, and attenuation coefficients, respectively. The fiber loss is compensated completely by an  optical amplifier. The dispersion is compensated at each receiver digitally.
			
		\begin{figure*}[!t]
	\centering
	
	{\includegraphics{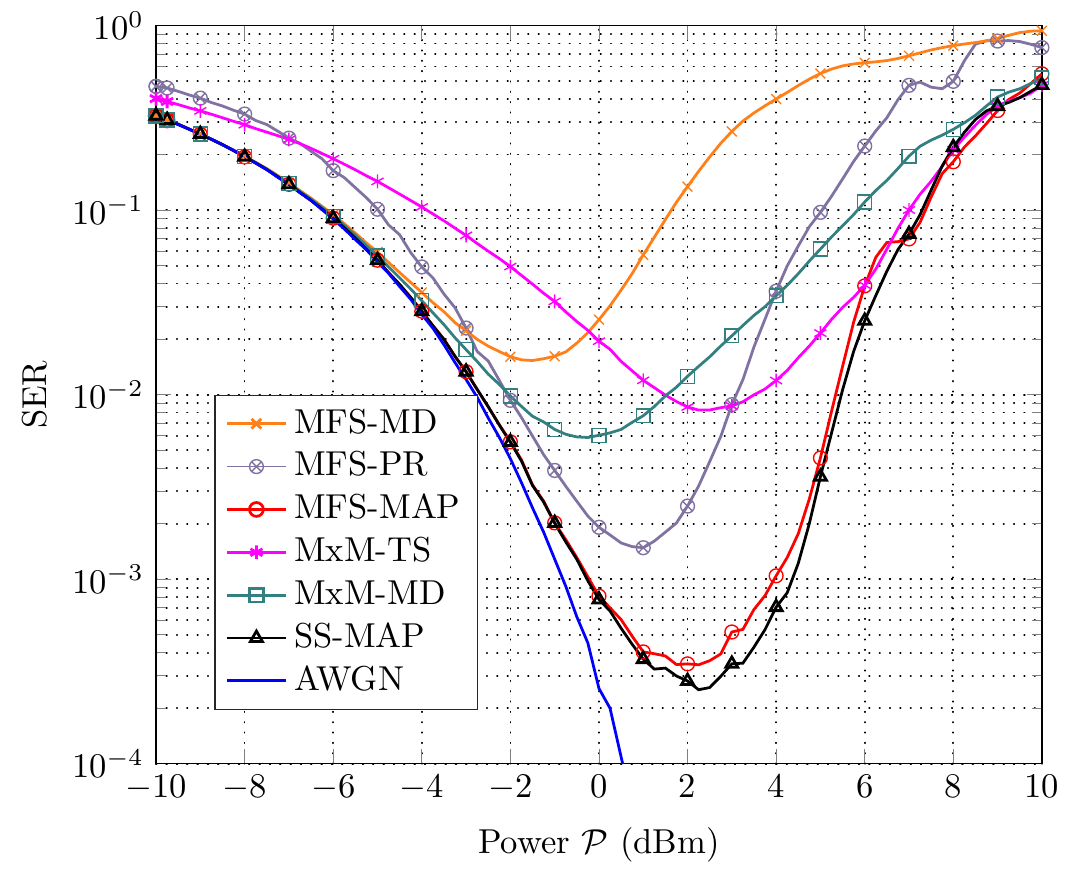}%
	}\hfil
	
	\caption{\small { The SER of    the six mismatched receivers introduced in Section~\ref{s3} is plotted for a single-span transmission with  low-dispersive fiber. The SER of an AWGN channel with the same noise variance   is also plotted for comparison.}
	}\label{errld}

\end{figure*}
			
			 We consider two fiber-optical systems with different dispersion parameters. The first  system deploys a quadruply clad fiber \cite[Ch.~1]{agrawal_2007_nfo} with $\beta_{2}=-1.27$ ps$^2$/km and the second system uses a standard SMF with $\beta_{2}=-21.7$ ps$^2$/km.  The channel spacing parameter is $\df=40 $ GHz. The values of the other parameters can be found in Table~\ref{t2}. The solution  of \eqref{ch:nlse} is approximated by  the split-step Fourier method \cite[Ch.~2.4.1]{agrawal_2007_nfo}; $100$ samples are taken from each symbol to discretize the input signal. Pulse shaping is the same as in Section~\ref{s4}. 

				\begin{figure*}[!t]
				\centering
				{\includegraphics{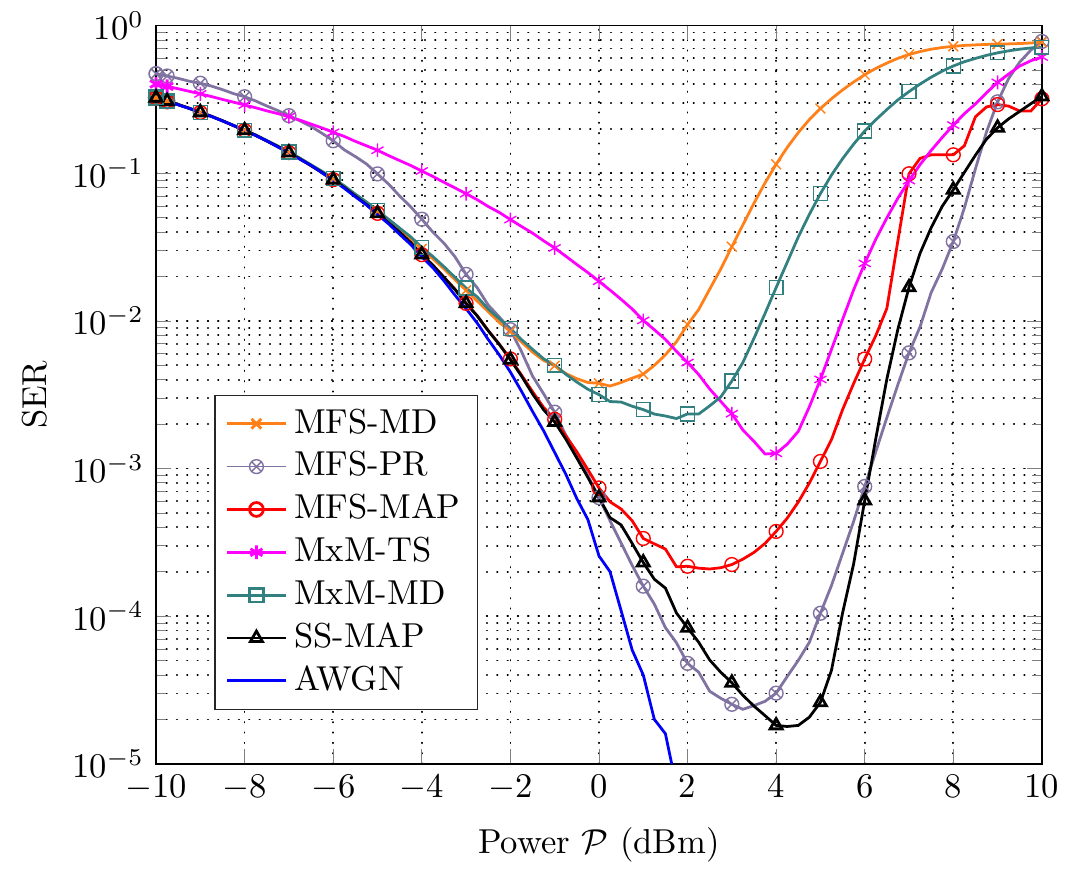}%
				}\hfil
				\caption{\small { The SER of the  six mismatched receivers introduced in Section~\ref{s3} is plotted for a single-span transmission with  standard SMF. The SER of an AWGN channel with the same noise variance   is also plotted for comparison.}
				}\label{errhd}	
			\end{figure*}

			Fig.~\ref{errld} illustrates the  performance of the six receivers for a low-dispersive fiber. It can be seen that because of the nonlinearity--dispersion interplay, the SER of all the receivers increases after reaching a global minimum.
			By using the SS-MAP and the MFS-MAP receivers, considerable performance gains can be achieved compared to MFS-PR. One can see that the MxM-MD and MxM-TS perform worse than MFS-PR but better than the MFS-MD receiver.  
			 Fig.~\ref{errhd} presents the SER  for a standard SMF. The dispersion is high and the  MFS-PS and SS-MAP perform better than the other receivers.  }

\section{ Conclusion and Discussion}\label{s5}

{Six} receivers were studied for a two-user simplified WDM channel and a novel demodulator, referred to as MxM, was proposed.
Our results indicate that the MFS-MD receiver, which is optimal for the AWGN channel, performs very poorly in the presence of optical nonlinear distortion. 
However, when the output of the  MFS is fed to a MAP detector, one can achieve performance close to the optimal receiver at low powers.
 In the high-power regime, the SER goes to zero with power for the optimal receiver as well as for the receivers based on the MxM demodulator. On the contrary, for receivers based on the MFS demodulator, the SER does not vanish.


In coherent optical transmissions the  signal spectrum broadens at high transmit power levels, because of the nonlinearity. 
The information embedded in the out-of-band frequencies is however ignored by the MFS demodulator.
 Our results  indicate that ignoring this information loss  deteriorates  performance substantially at high powers. 
  Moreover, by proposing the MxM demodulator, we showed that a vanishing SER can be obtained by a heuristic receiver that is simpler than the optimal one.
  
 {
 	When evaluated over a more realistic single-span fiber-optical channel, modeled by the NLS equation, the performance of all  receivers declines in the high-power regime. In the low-dispersion case two of the receivers analyzed in this paper, namely MFS-MAP and SS-MAP outperform the conventional MFS-PR receiver. Since the receivers in this paper were designed based on a simplified memoryless model, further improvement is expected by devising receivers that take into account both dispersion and nonlinearity. It seems that developing the optimal receiver in the presence of dispersion is a formidable task and heuristic methods should be considered.  A straightforward approach may be  optimizing the performance of the proposed receivers over different values of $\eta_1$. Since dispersion mitigates the effects of nonlinearity, the optimal $\eta_1$ may be smaller than the right-hand side of \eqref{eta}.
 	
 	Finally, we  note that equalization and phase recovery are  essential parts of today's optical receivers.  While the performance of the introduced receivers may be influenced by  these two steps,  we have not investigated the proper coupling of the equalization and the phase-recovery processes with the demodulation and detection steps. This is an interesting topic for   future studies. 
 	
 }


\end{document}